\documentclass[11pt]{article}

\usepackage{theorem}
\usepackage{fullpage}
\usepackage{ifthen}
\usepackage{verbatim}

\usepackage{amsmath,amssymb,graphicx,subfig}
\usepackage{caption}
\usepackage{natbib}
\usepackage{clrscode3e}

\usepackage{float}

\floatstyle{ruled}
\newfloat{algorithm}{thp}{lop}
\floatname{algorithm}{Algorithm}

\newcommand{\F}{\mathcal{F}}
\newcommand{\A}{\mathcal{A}}
\newcommand{\R}{\mathbb{R}}
\newcommand{\E}{\mathbb{E}}

\newcommand{\Fest}{\widehat{F}}
\newcommand{\Festm}{\widehat{F_m}}

\newcommand{\qest}{\widehat{F^{-1}}}
\newcommand{\qestm}{\widehat{F^{-1}_m}}

\newcommand{\wh}{\widehat}
\newcommand{\CH}{\mathcal{CH}}

\newcommand{\vmax}{H}

\newcommand{\Rmin}{\wh{R}_\text{min}}
\newcommand{\Rmax}{\wh{R}_\text{max}}
\newcommand{\Ralg}{R^{alg}}
\newcommand{\Ropt}{R^{opt}}

\newcommand{\I}{\mathcal{I}}

\newcommand{\M}{\mathcal{M}}
\newcommand{\dist}{F}
\newcommand{\eps}{\epsilon}

\newcommand{\val}{v}

\newcommand{\vv}{\varphi}

\newcommand{\ivv}{\overline{\varphi}}

\def\poly{\ensuremath{\mbox{poly}}}

\DeclareMathOperator*{\argmax}{arg\,max}

\newtheorem{theorem}	 			{Theorem}[section]
\newtheorem{lemma}		[theorem]	{Lemma}	
\newtheorem{fact}		[theorem]	{Fact}

\newtheorem{corollary}		[theorem]	{Corollary}
\newtheorem{prop}		[theorem]	{Proposition}
{\theorembodyfont{\rmfamily} }
{\theorembodyfont{\rmfamily} }
{\theorembodyfont{\rmfamily} \newtheorem{example}		[theorem]
	{Example}}
{\theorembodyfont{\rmfamily} }
{\theorembodyfont{\rmfamily} }
{\theorembodyfont{\rmfamily} }
{\theorembodyfont{\rmfamily} }
{\theorembodyfont{\rmfamily} }
\theoremstyle{break}
{\theorembodyfont{\rmfamily} }

\newenvironment{proof}{\noindent {\em {Proof:}}}{$\blacksquare$\vskip
	\belowdisplayskip}

\newenvironment{prevproof}[2]{\noindent {\em {Proof of
			{#1}~\ref{#2}:}}}{$\blacksquare$\vskip \belowdisplayskip}

\title{Ironing in the Dark}
\author{
	Tim Roughgarden\thanks{Department of Computer Science,
		Stanford University, 474 Gates Building, 353 Serra Mall, Stanford, CA 94305.
		This research was supported in part by NSF grant
		CCF-1215965.
		Email: {\tt tim@cs.stanford.edu}.} \and Okke Schrijvers\thanks{Department of Computer Science,
		Stanford University, 482 Gates Building, 353 Serra Mall, Stanford, CA 94305.
		This research was supported in part by NSF grant
		CCF-1215965.
		Email: {\tt okkes@cs.stanford.edu}.} 
}

\begin{document}

\maketitle

\begin{abstract}
This paper presents the first polynomial-time algorithm for position
and matroid auction environments that learns, from samples from  
an unknown bounded valuation distribution, an auction with expected
revenue arbitrarily close to the maximum possible.
In contrast to most previous work, our results apply to arbitrary (not
necessarily regular) distributions and the strongest possible
benchmark, the Myerson-optimal auction.
Learning a near-optimal auction for an irregular distribution is
technically challenging because it requires learning the appropriate
``ironed intervals,'' a delicate global property of the distribution.
\end{abstract}


\section{Introduction}

The traditional economic approach to revenue-maximizing auction
design, exemplified by Myerson~\cite{Myerson81},
posits a known prior distribution over what bidders are willing to pay,
and then solves for the auction that maximizes the seller's expected
revenue with respect to this distribution. 
Recently, there has been an explosion of work in computer science that
strives to make the classical theory more ``data-driven,''  replacing the
assumption of a known prior distribution with that of access to
relevant data, in the form of samples from an unknown distribution.  
In this paper, we study the problem of learning a near-optimal auction
from samples, adopting the formalism of Cole and Roughgarden \cite{Cole14}.
The idea of the model, inspired by PAC-learning~\cite{V84},
is to parameterize through samples the ``amount of knowledge'' that
the seller has about bidders' valuation distributions. 

We consider single-parameter settings, where each of $n$ bidders has a
private 
valuation (i.e., willingness to pay) for ``winning'' and valuation~0
for ``losing.''  
Feasible outcomes correspond to subsets of bidders
that can simultaneously win; the feasible subsets are known in
advance.\footnote{For 
  example, in auction with $k$
  copies of an item, where each bidder only wants one copy, feasible
  outcomes correspond to subsets of at most $k$ bidders.}
We assume that bidders' valuations are drawn i.i.d.\ from a
distribution~$F$ that is unknown to the seller.
However, we assume that the seller has access to $m$ i.i.d.\ samples
from the distribution~$F$ --- for example, bids that were observed in
comparable auctions in the past.
The goal is to design
a polynomial-time algorithm $A(v_1,\ldots,v_m)$ from samples
$v_i\sim F$ to
truthful auctions that, for every distribution $F$,
achieves expected revenue at least $1-\epsilon$ times the 
optimal expected revenue.\footnote{By a truthful auction, we mean one
  in which 
  truthful bidding is a dominant strategy for every bidder.  The
  restriction to dominant strategies is natural given our assumption
  of an unknown distribution.  Given this, the restriction to truthful
  auctions is without loss of generality (by the ``Revelation
  Principle,'' see~\cite{N07}).  Also, for the single-parameter
  problems that we study, there is a truthful optimal auction.}  
The sample complexity of achieving a given approximation
factor~$1-\eps$ is then the minimum number of samples~$m$ such that
there exists a learning algorithm $A$ with the desired approximation.
This model serves as potential ``sweet spot'' between worst-case and
average-case analysis, inheriting much of the robustness of the
worst-case model (since we demand guarantees for every underlying
distribution $\dist$) while allowing very good approximation guarantees.

\subsection{Our Results}

We give polynomial-time algorithms for learning a
$(1-\eps)$-approximate auction from samples, for arbitrary matroid
(or position)
auction environments and arbitrary (not necessary regular) distributions with
bounded valuations.  
(See Section~\ref{sec:prelim} for definitions.)  

Precisely, our main result is a polynomial-time algorithm that, given 
a matroid environment with $n$ bidders and $m$ i.i.d.\ samples from an
  arbitrary distribution~$F$ with maximum valuation $\vmax$,
  with probability $1-\delta$, 
  approximates the maximum-possible expected revenue
up to an
  additive loss of at most $3n\sqrt{\frac{\ln 2 \delta^{-1}}{2m}}\cdot 
\vmax$. 
Thus for every $\eps > 0$, the additive loss is at most $\eps$ (with
probability at least $1-\delta$) provided $m = 
\Omega(n^2\vmax^2\eps^{-2} \log \delta^{-1})$.
When valuations lie in $[1,\vmax]$, so that the optimal expected revenue
is bounded below, this result immediately implies the same sample
complexity bound for learning a $(1-\eps)$-(multiplicative)
approximate auction.  Our main result can also be used to give a
no-regret guarantee in a stochastic learning setting
(Section~\ref{sec:no-regret}).  

A lower bound of Cesa-Bianchi et al. \cite{Cesa-Bianchi13} implies that, already for
simpler settings, the quadratic dependence of our sample complexity bound on
$\tfrac{1}{\eps}$ is optimal.  A lower bound of Huang et al. \cite{HMR15} implies
that, already with a single bidder, the sample complexity must depend
polynomially on $\vmax$.  Whether or not the sample complexity needs
to depend on~$n$ is an interesting open question.

Our approach is based on a ``swiching trick''
(Proposition~\ref{prop:switching-trick})
and we believe it will lead to further applications.  A key idea is to
express the difference in expected revenue between an optimal and a
learned auction in a matroid or position environment purely in terms
of a difference in area under the true and estimated revenue curves.
This ``global'' analysis avoids having to compare explictly the true
and estimated virtual valuations or the optimal and learned allocation
rules.
With this approach, there is clear motivation behind each of the steps
of the learning algorithm, and the error analysis, while non-trivial,
remains tractable in settings more general than those studied in most
previous works.

\subsection{Why Irregular Distributions Are Interesting}

A major difference between this paper and much previous
work on learning near-optimal auctions from samples is that our positive
results accommodate {\em irregular distributions}.
To explain, we first note
that learning a near-optimal auction from samples 
is impossible without some restriction on the space of
possible valuation distributions~$F$.\footnote{To see this, consider
  a single-bidder problem and all distributions that take 
  on a value   $M^2$ with probability $\tfrac{1}{M}$ and 0 with
  probability $1-\tfrac{1}{M}$.  The optimal auction for such a
  distribution earns expected revenue at least $M$.  It is not difficult to
  prove that, for every $m$, there is no way to use $m$ samples
to achieve near-optimal revenue for every such distribution --- for
  sufficiently large $M$, all $m$ samples are~0 w.h.p.\ and the
  algorithm has to resort to an uneducated guess for $M$.}
This observation motivates imposing the weakest-possible conditions
under which non-trivial positive results are possible.

A majority of the literature on approximation guarantees for revenue
maximization (via learning algorithms or otherwise) restricts
attention to ``regular'' valuation distributions or subclasses
thereof; see related work below for examples and exceptions.  
Formally, a distribution $F$ with density~$f$ is {\em regular} if
\begin{equation}\label{eq:vv}
\vv(\val) = \val - \frac{1-\dist(\val)}{f(\val)}
\end{equation}
is a nondecreasing function of $\val$.  $\vv$ is also called the {\em
  virtual valuation} function.
Intuitively, regularity is a logconcavity-type assumption that
provides control over the tail of the distribution.
While many important distributions are regular,
plenty of natural distributions are not.
For example, Sivan and Syrgkanis \cite{Sivan13}
point out that mixtures of distributions (even of uniform
distributions) tend to be irregular, and yet obviously occur in the
real world.  This motivates relaxing regularity conditions while still
somehow excluding pathological distributions.
One interesting
approach is to parameterize a distribution by its ``degree of
irregularity;'' see \cite{HartlineXX,HMR15,Sivan13} for some proposals of
how to do this.  Another approach, and the one we take here, is to
bound the support of a distribution (to $[0,H]$ or $[1,H]$, say) and
allow it to be otherwise arbitrary.
Bounded but otherwise arbitrary valuation distributions clearly
capture interesting settings beyond those modeled by regular distributions.

\subsection{Why Irregular Distributions Are Hard}

To understand why irregular distributions (with bounded
valuations) are so much more technically challenging than regular
distributions, we need to review some classical optimal auction
theory.  We can illustrate the important points already in single-item
auctions. Myerson~\cite{Myerson81} proved that, for every regular
distribution~$F$, the optimal auction is simply a second-price auction
supplemented with a reserve price of~$\vv^{-1}(0)$, where~$\vv$
denotes the virtual valuation function in~\eqref{eq:vv}.  (The winner,
if any, pays the higher of the reserve price and the second-highest
bid.)  Thus, {\em learning the optimal auction reduces to learning the
  optimal reserve price}, a single statistic of the unknown
distribution.  
And indeed, for an unknown regular distribution~$F$,
there is a polynomial-time learning algorithm that
needs only $\poly(\tfrac{1}{\eps})$ samples to compute a
$(1-\eps)$-approximate auction~\cite{DRY10,HMR15}.


The technical challenge of irregular distributions is the need to {\em
  iron}.
When the virtual valuation function~$\vv$ of the distribution~$F$ is
not nondecreasing, Myerson~\cite{Myerson81} gave a recipe for
transforming $\vv$ into a nondecreasing ``ironed'' virtual valuation
function $\ivv$ such that the optimal single-item auction awards the
item to the bidder with the highest positive ironed virtual valuation
(if any), breaking ties randomly.
Intuitively, this ironing procedure identifies intervals of
non-monotonicity in $\vv$ and changes the value of the function to be
constant on each of these intervals.
(See also below and the exposition by Hartline \cite{HartlineXX}.)

The point is that the appropriate ironing intervals of a distribution
are a {\em global} property of the distribution and its (unironed)
virtual valuation function.  Estimating the virtual valuation function
at a single point --- all that is needed in the regular case --- would
appear much easier than estimating the right intervals to iron in the
irregular case.  

We present two examples to drive this point home.  The first, which is
standard, shows that not ironing can lead to a constant-factor loss in
expected revenue.  The second shows that tiny mistakes in the choice
of ironing intervals can lead to a large loss of expected
revenue.\footnote{See Appendix~\ref{sec:information} for a different
  example that demonstrates why irregular distributions are harder
  than regular ones.}
\begin{example}[Ironing Is Necessary for Near-Optimal Revenue]\label{ex:needtoiron}
The
 distribution is as 
  follows: with probability $1/H$ the value is $H$ (for a large $H$)
  and it is $1$ otherwise. The optimal auction irons the interval
  $[1,H)$ for expected revenue of $2-\frac{1}{n}$ \cite{HartlineXX},
  which approaches 2 with many bidders~$n$.  Auctions that do not
  implicitly or explicitly iron obtain expected revenue only~1.
\end{example}

\begin{example}[Small Mistakes Matter]\label{ex:matter}
Let $F$ be $5$ with probability $1/10$ and $1$ otherwise, and consider
a single-item auction with $10$ bidders. The optimal auction irons the
interval $[1,5)$ and has no reserve price. If there are at least two
  bidders with value $5$ one of them will get the item at price $5$;
  if all bidders have value $1$, one of them will receive it at price
  $1$. If there is exactly one bidder with value $5$, then her price
  is $\frac{1}{10}\cdot 1 + \frac{9}{10}\cdot 5 = \frac{46}{10}$. 
	
Now consider an algorithm that slightly \emph{overestimated} the end
of the ironing interval to be $[1,5+\epsilon)$ with $\epsilon>0$. 
(Imagine~$F$ actually has small but non-zero density above~5, so that
  this mistake could conceivably occur.)
Now
  all bids always fall in the ironing interval and therefore the item
  is always awarded to one of the players at price $1$. Not only do we
  lose revenue when there is exactly one bidder, but additionally we
  lose revenue for auctions with at least two bidders with value $5$.
This auction has even worse revenue than the standard second-price
auction, so the attempt to iron did more harm than good.
\end{example}
If we adapt Example~\ref{ex:matter} to slightly underestimate the ironing interval, there is almost no loss in revenue. This motivates an important idea in our learning algorithm, which is to closely protect against overestimation.

\subsection{Related Work}

Elkind~\cite{E07} gives a polynomial-time learning algorithm
for the restricted case of single-item auctions with discrete
distributions with known finite supports but with unknown
probabilities.  In the model
in~\cite{E07}, learning is done using an oracle that compares
the expected revenue of pairs of auctions, and $O(n^2K^2)$ oracles
calls suffice to determine the optimal auction (where $n$ is the
number of bidders and $K$ is the support size of the distributions).
Elkind~\cite{E07} notes that such oracle calls can be
implemented approximately by sampling (with high probability), but no
specific sample complexity bounds are stated.

Cole and Roughgarden \cite{Cole14} 
also give a polynomial-time algorithm for learning a
$(1-\eps)$-approximate auction for single-item auctions with possibly
non-identical bidders, under incomparable assumptions to~\cite{E07}:
valuation distributions that can be unbounded but must be strongly
regular. It is necessary and sufficient to have
$m=\poly(n,\frac{1}{\eps})$ samples, however in the analysis in
\cite{Cole14} the exponent in the upper bound is large (currently, 10).

The papers of Cesa-Bianchi et al.~\cite{CGM15} and Medina and
Mohri~\cite{MM14} give algorithms for learning the optimal
reserve-price-based single-item auction.  
Recall from Example~\ref{ex:needtoiron} that, with irregular
distributions, the best reserve-price-based auction can have expected
revenue far from optimal.

Dughmi et al.~\cite{dughmi2014sampling} proved negative
results (exponential sample complexity) for learning near-optimal
mechanisms in  multi-parameter settings that are much more complex
than the single-parameter settings studied here.  The paper also
contains positive results for restricted classes of mechanisms.

Huang et al. \cite{HMR15} give optimal sample complexity bounds for
the special case of a single bidder under several different
distributional assumptions, including 
for the case of bounded irregular distributions where they need $O(H\cdot \eps^{-2}\cdot\log(H\eps^{-1}))$ samples.

Morgenstern and Roughgarden \cite{MR15} 
recently gave general sample complexity
upper bounds which are similar to ours and cover all 
single-parameter settings (matroid and otherwise), although the
(brute-force) learning algorithms in \cite{MR15} are not
computationally efficient. 

For previously studied models about revenue-maximization with an unknown
  distribution, which differ in various respects from the model of Cole and Roughgarden
  \cite{Cole14}, see \cite{BBDSics11,Cesa-Bianchi13,KL2003focs}.
For other ways to parameterize partial knowledge about valuations,
see e.g.~\cite{ADMW13,CMZ12}.
For other uses of samples in auction design that differ from ours,
see Fu et al.~\cite{FHHK14}, who use samples to extend the 
Cr\'emer-McLean theorem~\cite{CM85}
to partially  known valuation distributions, and Chawla et
al.~\cite{Chawla14}, which is discussed further below.
For asymptotic optimality results in various symmetric 
settings (single-item auctions, digital goods), which identify
conditions under which the expected revenue of some auction of interest (e.g., second-price) approaches the
optimal with an increasing number of i.i.d.\ bidders, see
Neeman~\cite{N03},
Segal~\cite{segal},
Baliga and Vohra~\cite{BV03},
and Goldberg et al.~\cite{G+06}.
For applications of learning theory
concepts to prior-free auction design in unlimited-supply settings,
see Balcan et al.~\cite{BBHM08}.

Finally, the technical issue of ironing from samples comes up also in 
Ha and Hartline \cite{Ha13} 
and Chawla et al. \cite{Chawla14}, in models incomparable to the one
studied here.
The setting of \cite{Ha13} is constant-factor approximation guarantees
for prior-free revenue maximization, where
the goal is input-by-input rather than distribution-by-distribution
guarantees.  
Chawla et al. \cite{Chawla14} study non-truthful auctions, where bidders' true valuations need to be inferred from equilibrium bids, and aim to learn the optimal ``rank-based auction,'' which can have expected revenue a constant factor less than that of an optimal auction.
Our goal of obtaining a $(1-\eps)$-approximation of the maximum revenue
achieved by any auction is impossible in the settings
of~\cite{Ha13,Chawla14}.

Summarizing, this paper gives the first polynomial-time algorithm for
position and matroid environments that learns, from samples from 
an unknown irregular valuation distribution, an auction with expected
revenue arbitrarily close to the maximum possible.

\subsection{Organization}
Section~\ref{sec:prelim} covers learning and auction preliminaries
important for our analysis.  Section~\ref{sec:single-item} presents
many of our key ideas in the special case of a single-item auction.
Section~\ref{sec:environments} extends our results to position and matroid
auction environments.  Section~\ref{sec:no-regret} deduces a no-regret
guarantee from our learning algorithms.

\section{Preliminaries}
\label{sec:prelim}

\subsection{Empirical Cumulative Distribution Function and the DKW Inequality}
\label{sec:dkw}

Let $X=\{X_i\}_{i=1}^m$ be a set of $m$ samples, and let $X^{(i)}$ be
the $i^\text{th}$ order statistic. We use the standard notion of the
empirical cumulative distribution function (empirical CDF): 
$\wh{F_m}(v) =\frac{1}{m}\cdot|\{X_i:X_i \leq v\}|$.
The empirical CDF is an estimator for the quantile of a given
value. The Dvoretzky-Kiefer-Wolfowitz (DKW) inequality
\cite{Dvoretzky56,Massart90} states that the difference between the empirical CDF and the actual CDF decreases
quickly in the number of samples. Let $\epsilon_{m,\delta} =
\sqrt{\frac{\ln 2\delta^{-1}}{2m}}$, then
$\Pr\left[\sup_{v\in \mathbb{R}}\left|F(v)-\wh{F_m}(v)\right| \leq \epsilon_{m,\delta}\right]\geq 1-\delta$.
So the largest error in the empirical CDF shrinks as
$O(m^{-1/2})$.
For our purposes we will not need the CDF $F$, but rather its inverse $F^{-1}$. Define 
$\qestm(x)$ as
$X^{(\max(1,\lceil x\cdot m\rceil))}$ for $x \in [0,1]$.
(For convenience, define
$\qestm(x)$ as 0 if $x < 0$ and $\vmax$ if $x > 1$.)
By definition, for all $v\in[0,\vmax]$:
\begin{align}
\label{eq:F-1}
\qestm\left(\Festm(v)\right)\quad \leq\quad v \quad \leq \quad \qestm\left(\Festm(v)+\frac1m\right).
\end{align}
In the remainder of this paper, we will use $\wh{F}$, $\wh{F^{-1}}$,
and $\epsilon$ without explicitly referring to the number of samples
$m$ and confidence parameter $\delta$. 

\subsection{Optimal Auctions using the Revenue Curve}
\begin{figure}
	\centering
	\subfloat[Revenue curve $R(q)$ and the optimal induced revenue curve $R^{\star}(q)$.]{\includegraphics{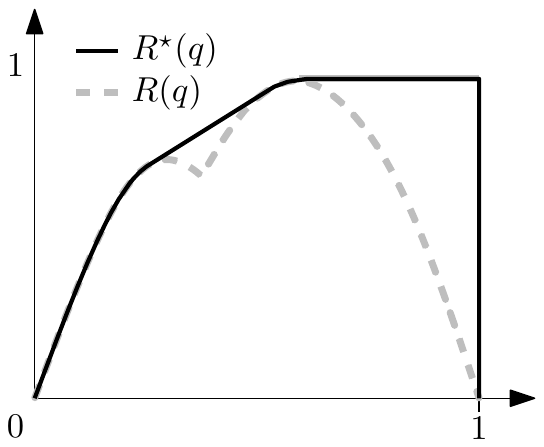}\label{fig:sub-rev-curve}}
	\quad
	\subfloat[Virtual value function $\vv(q)$ and the induced virtual value function $\vv^{\star}(q)$.]{\includegraphics{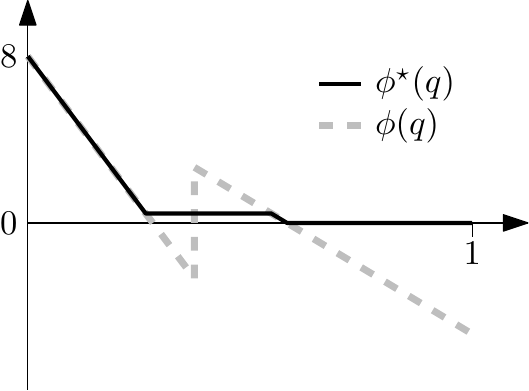}
		\label{fig:sub-virtual-value}}
	\caption{The dashed gray line is the revenue curve $R(q)$ and
          its derivative, the virtual value function $\vv(q)$, for a
          irregular bimodal distribution. In black we have the curves
          $\vv^{opt}$ and $R^{opt}$ corresponding to the optimal
          ironing interval, and optimal reserve price for this
          distribution. This example taken from \cite[Chapter
            3]{HartlineXX}.} 
	\label{fig:rev-curve}
\end{figure}


\paragraph{Revenue Curve and Quantile Space} 
For a value~$v$ and a distribution~$F$, we define the corresponding
{\em quantile} as $1-F(v)$.\footnote{This is for consistency with
  recent literature; ``reversed quantile'' might be a more accurate term.}
The quantile of $v$ corresponds to the sale probability for a single
bidder with valuation drawn from $F$ as the posted price~$v$.
Note that higher values correspond to lower quantiles.
The \emph{revenue curve}
$R$ for value distribution $F$ is the curve, in quantile space, that
is defined as $R(q) = q\cdot F^{-1}(1-q)$, see the dashed curve in
Figure~\ref{fig:sub-rev-curve}. 
Myerson \cite{Myerson81} showed that the ex ante expected revenue for a bidder is
$\E_{q\sim U[0,1]}[-R(q)\cdot y'(q)]$, where $y$ is the interim allocation function for the standard Vickrey auction (defined formally below).
The derivative of $R$ is the virtual
value function $\vv(.)$ --- recall~\eqref{eq:vv} --- in quantile space, see
Figure~\ref{fig:sub-virtual-value}.
For regular distributions 
the
revenue curve is concave and the virtual value function
non-increasing, but for irregular distributions this is not the case. 

Myerson \cite{Myerson81} showed that for any non-concave revenue curve $R$,
one can create an allocation rule that will yield the same revenue as
$R$'s convex hull $\CH(R)$. This procedure is called ironing, and for
each interval where $R$ differs from $\CH(R)$, we take the virtual
value to be slope of the convex hull over this interval. This means
the virtual values are constant, and hence any two bids in that range
are interchangeable, so the ironed allocation function is constant on
this interval.\footnote{It takes some effort to show that keeping the
  allocation probability constant on an interval has exactly the
  effect we described here \cite{Myerson81}.} The resulting
\emph{ironed} revenue curve $R^\star$ will be concave and the 
corresponding ironed virtual
value function $\vv^\star$ will be monotonically non-increasing. 
It is also useful to think of a reserve price~$r$ (with corresponding
quantile 
$q_r$) in a similar way, as effectively changing the virtual valuation
function so that
$\vv^{\star}(q)=0$ whenever $q\geq q_r$
(Figure~\ref{fig:sub-virtual-value}), 
with the corresponding revenue curve $R^{\star}$ 
constant in that region (Figure~\ref{fig:sub-rev-curve}).\footnote{Most of the existing literature would not consider the effect of the reserve price on the revenue curve, in which case the black and dashed lines would coincide after the second peak. However, by including its effect as we did, we'll be able to apply the Switching Trick described below.}

More generally, given any set
of disjoint ironing intervals $\I$ and reserve price
$r$, both in value space, we can imagine these modifying the revenue
curve as follows.  (For now this is a thought experiment;
Proposition~\ref{prop:switching-trick} connects these revenue curve
modifications back to allocation rule modifications.)
Let $R$ be the revenue curve without ironing
or a reserve price, and define $R^{(\I,r)}$ as the revenue curve 
\emph{induced} by a set $\I$ of ironing intervals and reserve price
$r$.
This curve is defined by
\begin{align}
R^{(\I,r)}(q) = \begin{cases}
R(F(r)) & \text{if }q>F(r)\\
R(q_a) + \frac{q-q_a}{q_b-q_a}\left(R(q_b)-R(q_a)\right) & \text{if
}q\in(q_{a},q_{b}]\text{ with } [F^{-1}(1-q_{b}),F^{-1}(1-q_{a}))\in \I\\ 
R(q) & \text{otherwise}.
\end{cases}\label{eq:induced-R}
\end{align}

Given $\I$ and $r$ as above, we define the auction $A^{(\I,r)}$ as
follows: given a bid profile: (i) reject every bid with $b_i < r$;
(ii) for each ironing interval $[a,b)\in \I$, 
treat all bids $\{b_j :   a \leq b_j < b\}$ as identical (equal to
some common number between $a$ and $b$);
(iii) among the remaining bidders, maximize the social welfare (sum of
the ironed bids of the winners); (iv) charge payments so that losers
always pay~0 
and so that truthful bidding is a dominant strategy for every player.
This auction is well defined (i.e., independent of the choice of the
common numbers in (ii)) in settings where the computation in~(iii)
depends only on the ordering of the ironed bids, and not on their
numerical values.  In this case, the payments in~(iv) are uniquely
defined (by standard mechanism design results).
This is the case in every matroid
environment\footnote{In a matroid environment, the set $\F$ of
  feasible outcomes satisfies: (i) (downward-closed) $T \in \F$ and $S
  \subseteq T$ implies $S \in F$; and (ii) (exchange property)
whenever $S,T \in \I$ with $|T| < |S|$, there is
some $i \in S \setminus T$ such that $T \cup \{i\} \in \I$.} and also in
position auctions (see~\cite{EOS07,V07}).
In such a setting, we use $\A$ to denote the set of all auctions of
the form $A^{(\I,r)}$.
We restrict attention to such settings for the rest of the paper.

\paragraph{The Switching Trick}
Given a distribution $F$,
we explained two ways to use ironing intervals~$\I$ and a reserve
price~$r$: (i) to define a modified revenue curve $R^{(\I,r)}$ (and
hence virtual valuations); 
or (ii) to define an auction $A^{(\I,r)}$.  The ``switching trick''
formalizes the connection between them: the expected 
virtual welfare of the
welfare-maximizing auction with the modified virtual valuations
(corresponding to the derivative of $R^{(\I,r)}$) equals 
the expected virtual welfare of the modified auction $A^{(\I,r)}$ with
the original virtual valuations.

More formally,
let $x_i : \R_+^n \rightarrow \R_+$ be the ex-post allocation function of the welfare maximizing truthful auction
that takes the bids ${\bf b}$ of all players and results in the
allocation to bidder $i$. The \emph{interim} allocation function $y_i
: [0,1] \rightarrow \R_+$ is the expected allocation to bidder $i$
when her quantile is $q$, where the expectation is over the quantiles
of the other bidders: $y_i(q_i) = \E_{{\bf q}_{-i}\sim U[0,1]^{n-1}}
[x(F^{-1}(1-q), F^{-1}({\bf 1}-{\bf q}_{-i}))]$ where $F^{-1}({\bf 1}-{\bf q}_{-i})$ is ${\bf b}_{-i}$
for which each $b_j = F(1-q_j)$.  
%
For example, in the standard Vickrey (single-item) auction with $n$
bidders, every 
bidder~$i$ has the interim allocation function $y_i(q) =
(1-q)^{n-1}$.\footnote{In general matroid settings,
  different bidders can have different interim allocation functions
  (even though valuations are i.i.d.)}

For every auction of the form $A^{(\I,r)}$,
the interim allocation function
$y_i^{(\I,r)}$ of a bidder~$i$
can be expressed in terms of the interim allocation
function $y_i$ without ironing and reserve price (see also
Figure~\ref{fig:sub-virtual-value}): 
\begin{align}
y_i^{(\I,r)}(q) = \begin{cases}
0 & \text{if }q>F(r)\\
\frac{1}{q_b - q_a}\int_{q_a}^{q_b} y(q) dq & \text{if
}q\in[q_a,q_b)\text{ with } [F^{-1}(1-q_b),F^{-1}(1-q_a))\in \I\\
y_i(q) & \text{otherwise}.
\end{cases}\label{eq:induced-y}
\end{align}


\begin{prop}[Switching Trick]
\label{prop:switching-trick}
Consider a matroid or position auction setting, as above.  For every
valuation distribution $F$, every reserve price $r$, every set $\I$ of
disjoint ironing intervals, and every bidder $i$,
\begin{align*}
\E_{q\sim U[0,1]}[R(q)\cdot (y_i^{(\I,r)})'(q)]
= \E_{q\sim U[0,1]}[R^{(\I,r)}(q)\cdot y_i'(q)].
\end{align*}
\end{prop}
\begin{proof}
	Fix $F$, $\I$, $r$, and $y$. Let $y^{(\I,r)}$ be the interim
	allocation rule from running auction $A_{\I,r}$. Let $R$ be the
	revenue curve of $F$ and let $R^{(\I,r)}$ denote the revenue curve
	induced by $\I$ and $r$. 
	
	\begin{itemize}		
		
		\item Define a distribution $F^{(\I,r)}$ (which is not equal to $F$
		unless $\I=\emptyset$ and $r=0$) that has the property that its
		revenue curve $q\cdot F^{(\I,r)}(1-q)$ is $R^{(\I,r)}$. To see that
		this is well-defined, observe the following. Any line $\ell$ through
		the origin only intersects $R^{(\I,r)}$ once (if there are point
		masses in $F$ then a line through the origin intersect $R$ in a
		single interval). This means that we can use $R^{(\I,r)}$ to
		construct $F^{(\I,r)}$: $F^{(\I,r)}(v)$ is the $q$ for which $q\cdot
		v$ intersects with $R^{(\I,r)}(q)$ (if there are any point masses
		then there will be a range of $q$ for which this is the case; in
		that case take the largest such $q$). 
		Alternatively, see Hartline and Roughgarden \cite{HR08} for an explicit formula for $F^{(\I,r)}$.
		
		\item If we run the same auction $A_{\I,r}$ on bidders with values drawn
		from $F^{(\I,r)}$, the expected revenue is identical to the auction
		with bidder values drawn from $F$: $$\E_{q\sim U[0,1]}[R(q)\cdot
		\left(y^{(\I,r)}\right)'(q)]= \E_{q\sim
			U[0,1]}[R^{(\I,r)}(q)\cdot \left(y^{(\I,r)}\right)'(q)].$$ This
		can easily be seen by filling in the definitions from
		\eqref{eq:induced-R} and \eqref{eq:induced-y}.  
		
		\item If the bidders have distribution $F^{(\I,r)}$, then we might as
		well not iron or have a reserve price at all; so $$\E_{q\sim
			U[0,1]}[R^{(\I,r)}(q)\cdot
		\left(y^{(\I,r)}\right)'(q)]=\E_{q\sim
			U[0,1]}[R^{(\I,r)}(q)\cdot y'(q)].$$ This is also easily seen by
		filling in the definitions. 
		
	\end{itemize}
\end{proof}


\subsection{Notation}
In the remainder of this paper, our analysis will rely on bounding the difference in revenue of an auction with respect to the optimal auction in terms of their revenue curves. We will use the following conventions, see Table~\ref{tbl:notation}. The unaltered revenue curve for distribution is denoted by $R(q) = q\cdot F^{-1}(1-q)$. To denote when we use an estimator for a revenue, i.e. a revenue curve that is constructed based on samples, we use a hat: $\widehat{R}(q) = q\cdot \qest(1-q)$. Based on the available samples we construct high-probability upper and lower bounds for $R$, that are thus denoted as $\Rmax(q) = q\cdot \qest(1-q+\eps+\frac{1}{m})$ and $\Rmin(q) = q\cdot \qest(1-q-\eps)$.

\begin{table}
\centering
\begin{tabular}{|l|l|}
	\hline
	Revenue Curve & Description \\
	\hline\hline
	$R$ & $q\cdot F^{-1}(1-q)$\\
	$\Rmin$ & $q\cdot \qest(1-q-\eps)$\\
	$\Rmax$ & $q\cdot \qest(1-q+\eps+\frac1m)$\\
	\hline
\end{tabular}
\caption{Overview of notation for revenue curves.}
\label{tbl:notation}
\end{table}

We use superscript to denote when a revenue curve is ironed and has a reserve price. For a general set of ironing intervals $\I$ and reserve price $r$, $R^{(\I,r)}$ is the revenue curve induced by it, see \eqref{eq:induced-R}. The superscript $\star$ denotes that the revenue curve is optimally ironed and reserved, i.e. $R^\star$ is the revenue curve of Myerson's auction using $F$, and $\Rmax^\star$ is the revenue curve corresponding to the convex hull of $\Rmax$ that additionally stays constant after the highest point. Finally, we'll use $R^{alg}$ and $R^{opt}$ to denote an algorithm ALG's revenue curve and the optimal revenue curve for $F$ respectively (thus $R^{opt}=R^\star$, but we'll use $R^{opt}$ to emphasize its relation to $R^{alg}$).

For the ironing intervals $\I$ (and reserve price $r$) we use $\I_q$ (resp. $r_q$) when it is important that the ironing intervals are defined in \emph{quantile space}. Finally, $\I_{opt}$, $\I_{alg}$, and $\I_{max}$ (and similarly for reserve $r$) refer to the ironing intervals of the optimal auction, algorithm ALG and the optimal ironing intervals for $\Rmax$ respectively.

\section{Additive Loss in Revenue for Single-Item Auctions}
\label{sec:single-item}

In this section we describe an algorithm that takes a set $X$ of $m$
samples, and a confidence parameter $\delta$ as input, and outputs a
set $\I$ of ironing intervals and a reserve price $r$, both in value
space. 
We focus on the case where $\I$ and $r$ are used in a
single-item auction $A_{(1)}^{\I,r} \in \A$ (recall the notation in
Section~\ref{sec:prelim}) and show that the
additive loss in revenue of $A_{(1)}^{\I,r}$ with respect to the
revenue of the optimal auction $A^{opt}_{(1)}$ for single-item
auctions is $O(\epsilon \cdot n \cdot \vmax)$, with
$\epsilon=\sqrt{\frac{\ln 2\delta^{-1}}{2m}}$. In section~\ref{sec:environments} we extend the results to matroid and position auctions.

\begin{theorem}[Main Theorem]
	\label{thm:main-theorem}
	For a single-parameter environment with optimal auction of the form
        $A^{(\I,r)}$ with $n$ i.i.d.\ bidders with values from unknown
        irregular distribution $F$, $m$ i.i.d.\ samples from $F$, with
        probability $1-\delta$, the additive loss in expected revenue
        of Algorithm~\ref{alg:em} compared to the optimal expected
        revenue is at most $3\cdot n\cdot \vmax\cdot \sqrt{\frac{\ln
            2\delta^{-1}}{2m}}$. 
\end{theorem}

\subsection{The Empirical Myerson Auction}

We run the Empirical Myerson auction
(Algorithm~\ref{alg:em}), which we divided this into two parts: the first
is an algorithm $ALG$ (Algorithm~\ref{alg:ca}) that computes the
ironing intervals $\I$ and reserve price $r$ based on samples $X$ and
confidence parameter $\delta$.
The second step is to run the welfare-maximizing auction
subject to ironing and reservation. 
In this section we focus on analyzing the single-item auction, but the only place this is used is in line 2 of Algorithm~\ref{alg:em}. Auctions for position auctions and matroid environments are defined completely analogously.

\begin{algorithm}[t]
	\begin{codebox}
		\Procname{$\proc{ComputeAuction}(X, \delta)$}
		\li Construct $\qest$ from $X$; let $\epsilon=\sqrt{\frac{\ln 2 \cdot \delta^{-1}}{2|X|}}$.
		\li Construct $\wh{R}_\text{min}(q) = q\cdot\qest(1-q-\epsilon)$.
		\li Compute the convex hull $\CH(\Rmin)$, of $\Rmin$.
		\li Let $\I_q$ be the set of intervals where $\Rmin$ and $\CH(\Rmin)$ differ.
		\li \For \kw{each} quantile ironing interval $(a_i, b_i) \in \I_q$ \Do
		\li Add $[\qest(1-b_i-\epsilon),\ \qest(1-a_i-\epsilon))$ to $\I$.
		\End
		\li Let the reserve quantile be $r_q = \argmax_q \Rmin(q)$.
		\li Let the reserve price be $r = \qest(1-r_q-\epsilon)$.
		\li \Return $(\I,r)$
	\end{codebox}
	\caption{Compute the ironing intervals $\I$ and reserve price $r$.}
	\label{alg:ca}	
\end{algorithm}

\begin{algorithm}[t]
	\begin{codebox}
		\Procname{$\proc{EmpiricalMyerson}(X, \delta,{\bf b})$}
		\li $\I, r \leftarrow \proc{ComputeAuction}(X,\delta)$
		\li \Return $A_{(1)}^{\I,r}({\bf b})$
	\end{codebox}
	\caption{Empirical Myerson.}
	\label{alg:em}	
\end{algorithm}


The Empirical Myerson auction takes an estimator for the quantile
function $\qest$ and constructs its revenue curve.\footnote{There are
  variations that differ from this approach, most notably the Random
  Sampling Empirical Myerson (RSEM) \cite{Devanur14} which constructs
  the revenue curve from the bidders in the auction, rather than from
  past data. This leads to constant-factor approximation bounds for
  prior-free auctions.} From this,
the convex hull $\CH(R)$ is computed and wherever $\CH(R)$ and $R$
disagree, an ironing interval is placed. Then, the highest point on
$R$ is used to obtain the reserve price quantile $q_r = \argmax_q
R(q)$. Note that this is all done in quantile space, but we need to
specify the reserve price and ironing intervals in value space. So the
last step is to use the empirical CDF $\Fest$ to obtain the values at
which to place the reserve price and ironing intervals. 

Algorithm~\ref{alg:em} follows that approach, with the exception that
in line 2 of \proc{ComputeAuction}, we take the empirical quantile
function to be $\qest(1-q-\epsilon)$ rather than the arguably more
natural choice of $\qest(1-q)$. 
By the DKW inequality, $\qest(1-q-\epsilon)\leq F^{-1}(1-q)$ with
probability $1-\delta$ (we prove this in Lemma~\ref{lem:bounds-on-R}).
The motivation here is to protect against overestimation
--- recall the cautionary tale of
Example~\ref{ex:matter}. That this approach indeed
leads to good revenue guarantees is shown in this section. 

\subsection{Additive Revenue Loss in Terms of Revenue Curves}

We start with a technical lemma that reduces 
bounding the loss in revenue to
bounding the estimation error due to using samples as opposed to the true
distribution $F$. 

\begin{lemma}
	\label{lem:tech-lemma}
	For a distribution $F$, let $\Ralg$ be the revenue curve induced by an algorithm $ALG\in \A$ and let $\Ropt$ be the optimal induced revenue curve. The additive revenue loss of $ALG$ with respect to $OPT$ is at most:
	\begin{align*}
		n\cdot \max_{q\in[0,1]} \left(\Ropt(q) - \Ralg(q)\right).
	\end{align*}
\end{lemma}
\begin{proof}
	First, to calculate the ex ante expected revenue of a bidder $i$ with revenue
	curve $R$, ironing intervals $\I$ and reserve price $r$, we have by
	Myerson \cite{Myerson81}: 
	\begin{align}
		Rev[R,\I,r] &= -\E_{q\sim U[0,1]}[R(q)\cdot (y^{(\I,r)})'(q)].\label{eq:rev}
	\end{align}
	Next, we apply the Switching Trick of
	Proposition~\ref{prop:switching-trick}. Let $\I_{alg}, \I_{opt}$ be
	the sets of ironing intervals and $r_{alg},r_{opt}$ be the reserve
	prices of $ALG$ and $OPT$ respectively. This yields total revenues: 
	\begin{align*}
		\frac{Rev[F,I_{alg},r_{alg}]}{n} &= -\int_0^1 \Ralg(q)y'(q) dq,&
		\frac{Rev[F,I_{opt},r_{opt}]}{n} &= -\int_0^1 \Ropt(q)y'(q) dq.
	\end{align*}	
	Note that the interim allocation function $y$ in both cases is the
	same one; the only difference between $y^{(\I_{opt},r_{opt})}$ and
	$y^{(\I_{alg},r_{alg})}$ was the ironing intervals and reserve price,
	so after applying the switching trick, $y$ is simply the truthful
	welfare-maximizing interim allocation rule (for single-item auctions
	this is the probability for bidder $i$ to have the largest value $v_i
	= \max_j v_j$). This is the key point in our analysis, and it allows
	us to compare the expected revenue of both auctions directly: 
	\begin{align*}
			Rev[F,\I_{opt},r_{opt}] - Rev[F,\I_{alg},r_{alg}] &= n\left( -\int_0^1 \Ropt(q)y'(q) dq +\int_0^1 \Ralg(q)y'(q) dq\right)\\
		&= n\int_0^1 \left(\Ropt(q) - \Ralg(q) \right)\left(-y'(q)\right) dq\\
		&\leq n\cdot\max_{q\in[0,1]}\left(\Ropt(q) - \Ralg(q) \right) \cdot  \int_0^1 -y'(q) dq\\
		&= n\cdot \max_{q\in[0,1]}\left(\Ropt(q) - \Ralg(q) \right) \cdot \left(-y(1) + y(0)\right)\\
		&= n\cdot\max_{q\in[0,1]}\left(\Ropt(q) - \Ralg(q) \right).
	\end{align*}
	The inequality holds as $-y'$ is non-negative. The last equality holds
	as the probability of winning when your quantile is $0$ is $y(0)=1$, and the
	probability of winning when your quantile is $1$ is
	$y(1)=0$. Rearranging the terms yields the claim. 
\end{proof}

This is significant progress:
the additive loss in revenue can be bounded in terms of the
induced revenue curves of an algorithm ALG and the optimal algorithm,
two objects that we have some hope of getting a handle on.
Of course, we still need to show that the ironed revenue curve of
Algorithm~\ref{alg:em} is pointwise close to the ironed revenue curve
induced by the optimal auction (Section~\ref{sec:rev-error}).

\subsection{Bounding the Error in the Revenue Curve}
\label{sec:rev-error}

We implement the following steps to prove that
the error in the learning algorithm's estimation of the revenue curve
is small. 

\begin{itemize}
	\item (Lemma~\ref{lem:bounds-on-R}) We show that we can sandwich the actual revenue curve (without ironing or reserve price) $R$ between two empirical revenue curves, $\Rmin$ and $\Rmax$ that are defined using the empirical quantile function.
	\item (Lemma~\ref{lem:opt-ub} and Lemma~\ref{lem:alg-lb}) Let $\Rmax^\star$ (resp. $\Rmin^\star$) be the optimally induced revenue curve for $\Rmax$ (resp. $\Rmin$). The revenue curve induced by Algorithm~\ref{alg:em}, $\Ralg$, is pointwise higher than the optimal induced revenue curve of the lower bound $\Rmin^\star$, and the optimal induced revenue curve for the upper bound, $\Rmax^\star$, is pointwise higher than $\Ropt$.
	\item (Lemma~\ref{lem:loss-rev-curve}) Finally, we show that $\Rmax^\star(q) - \Rmin^\star(q)$ is small for all $q$, and therefore the additive loss is small.
\end{itemize}


\begin{figure}
	\centering
	\includegraphics[width=.45\textwidth]{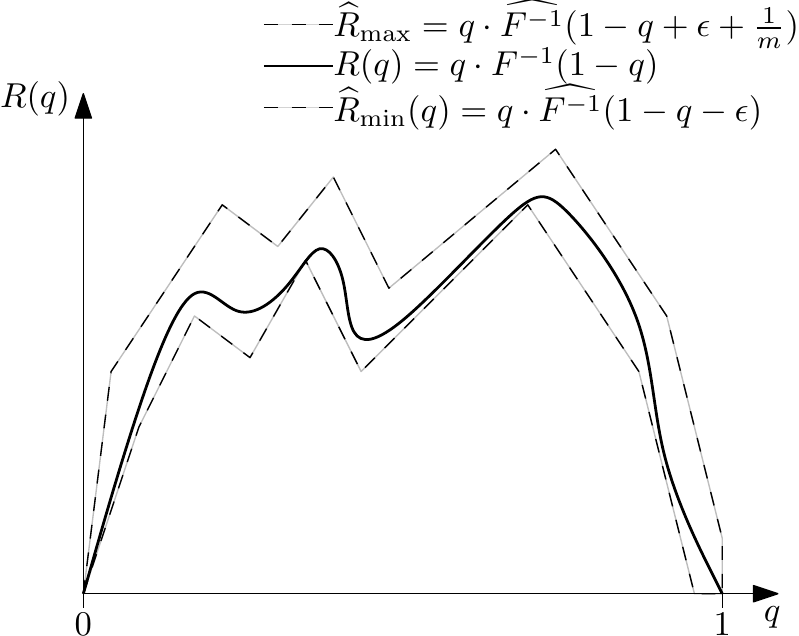}
	\caption{We can sandwich $R$ between two estimated revenue curves $\Rmin$ and $\Rmax$.}
	\label{fig:sub-rmax-r-rmin}
\end{figure}

\begin{lemma}
	\label{lem:bounds-on-R}
	Let $\Rmin(q) = q\cdot\qest(1-q-\epsilon)$ and $\Rmax(q) = q \cdot \qest(1-q+\epsilon+\frac1m)$. With probability at least $1-\delta$ for all $q\in[0,1]$:
	\begin{align*}
	\Rmin(q)\quad \leq \quad R(q) \quad \leq \quad \Rmax(q).
	\end{align*}
\end{lemma}
\begin{proof}
	See Figure~\ref{fig:sub-rmax-r-rmin} for graphical intuition. We will start with the first inequality: $\Rmin(q)\leq R(q)$. By the DKW inequality with probability $1-\delta$ the following holds for all $v$:
	\begin{align}
		\label{eq:dkw}
		-\epsilon \leq \Fest(v) - F(v) \leq \epsilon
	\end{align}
	Define $q' = F(v)$. We can rewrite the first inequality of \eqref{eq:dkw}:
	\begin{align*}
		q' - \epsilon &\leq \Fest(v)\\
		\qest\left(q' - \epsilon\right) &\leq \qest\left( \Fest(v)\right)
	\end{align*}
	Since $\qest$ is mononotonically non-decreasing.
	We now invoke \eqref{eq:F-1}: 
	\begin{align*}
		\qest\left(q' - \epsilon\right) \leq \qest\left( \Fest(v)\right) \leq v = F^{-1}(q').
	\end{align*}
	Now let $q=1-q'$, and multiply both sides by $q$ to obtain: $\Rmin(q)\leq R(q)$.
	
	The proof for the upper bound of $R$ is analogous with the exception that we pick up another $\frac1m$ term to invoke \eqref{eq:F-1}. 
\end{proof}

So while the the algorithm does not know the exact revenue curve $R$, it
can be upper bounded by $\Rmax$ and lower bounded by $\Rmin$. We'll
use $\Rmin$ to give a lower bound on the revenue curve $\Ralg$ induced
by Algorithm~\ref{alg:em},
and $\Rmax$ to give an upper
bound on the revenue curve $\Ropt$ induced by Myerson's optimal
auction. We start with the latter. 

\begin{lemma}
	\label{lem:opt-ub}
	Let $\Ropt$ be the optimal induced revenue curve of $R$, and let $\Rmax^\star$ be the optimal induced revenue curve for $\Rmax$. Then with probability $1-\delta$ for all $q\in[0,1]$:
	\begin{align*}
	\Ropt(q)\quad \leq \quad \Rmax^\star(q).
	\end{align*}
\end{lemma}
\begin{proof}
	By Lemma~\ref{lem:bounds-on-R} we know that for every $q$, $\Rmax(q) \geq R(q)$. First take the effect of optimally ironing both curves into account. Optimal ironing will lead to induced revenue curves $\CH(\Rmax)$ and $\CH(R)$ and since $\Rmax$ is pointwise higher, so is its convex hull.
	What remains to be proven is that this property is maintained after the effect of the reserve price of the curve. The reserve quantile is the highest point on the curve, and the induced revenue curve will stay constant for all higher quantiles. Let $q$ be the reserve quantile for $R$. If the reserve quantile for $\Rmax$ happens to be the same, then the statement is proven. Moreover, if the optimal reserve quantile for $\Rmax$ is not $q$, it can only be higher, in which case the statement is also proven. 
\end{proof}

So $\Rmax^\star$ is pointwise higher than $\Ropt$. Proving that $\Rmin^\star$ is a lower bound for $\Ralg$ is slightly more involved since the ironing intervals and reserve price are given by Algorithm~\ref{alg:em}, and may not be optimal. Therefore, the induced revenue curve $\Ralg$ is in general not concave, and the reserve quantile may not be at the highest point of the curve. However, since $\Ralg$ is induced by the reserve price and ironing intervals that were chosen by looking at $\Rmin$, there is enough structural information to prove the lower bound.

\begin{lemma}
	\label{lem:alg-lb}
	Let $\Ralg$ be the revenue curve induced by Algorithm~\ref{alg:em} and let $\Rmin^\star$ be the optimal induced revenue curve for $\Rmin$. Then with probability $1-\delta$ for all $q\in[0,1]$:
	\begin{align*}
	\Rmin^\star(q)\quad \leq \quad \Ralg(q).
	\end{align*}
\end{lemma}
\begin{proof}
	Again we first argue that after ironing the induced curve $\Rmin^\star$ is completely below $\Ralg$, and subsequently show that this is still true after setting the reserve price.
	
	Algorithm~\ref{alg:em} picks the ironing intervals based on where $\Rmin^\star$ differs from $\CH(\Rmin)$. However, $\Rmin$ and $\Ralg$ are defined in quantile space, whereas the ironing intervals are defined in value space. This means that $\Ralg$ is not necessarily ironed on the same intervals in quantile space as $\Rmin^\star$.
	The line that goes through the origin and intersects $\Rmin$ at $q$, will intersect $R$ at the point $q'$ where the $F^{-1}(q') = \qest(1-q-\epsilon)$, see Figure~\ref{fig:chrmin}.\footnote{It's important to note that we do not need to know $F^{-1}$, to predict the effect of ironing the actual revenue curve.} Any line through the origin intersects $\Rmin$ before it intersects $R$, and since ironing is done according to these lines through the origin, this property is maintained after ironing. 
	The reserve price is determined in the same way, and hence $\Ralg(q) \geq \Rmin^\star(q)$ for all $q\in[0,1]$ as claimed.
\end{proof}

\begin{figure}
	\centering
	\includegraphics[scale=.9]{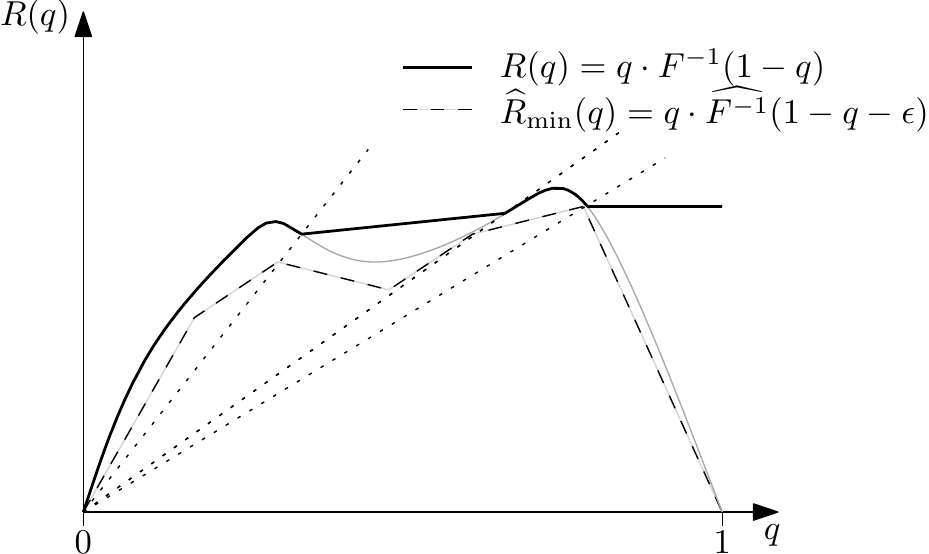}
	\caption{When we pick a value range to iron based on $\Rmin$, its effect on the actual revenue curve can be seen. The quantiles of the start and end point of the ironing procedure are given by the line that intersects the start and end point on $\Rmin$ and the origin.}
	\label{fig:chrmin}
\end{figure}

So we have our upper bound and lower bounds in terms of $\qest$. Finally we show that the difference between the two is small.

\begin{lemma}
	\label{lem:loss-rev-curve}
	For a distribution $F$, let $\Ralg$ be the revenue curve induced by Algorithm~\ref{alg:em}, $\Ropt$ be the optimal induced revenue curve. Using $m$ samples from $F$, with probability $1-\delta$ and $\epsilon = \sqrt{\frac{\ln 2\cdot \delta^{-1}}{2m}}$ for all $q$: 
	\begin{align*}
	\Ropt(q) - \Ralg(q) \leq  \Rmax^\star(q) - \Rmin^\star(q) \leq \left(2\epsilon+\frac1m\right)\vmax \leq 3\epsilon\cdot \vmax.
	\end{align*}
\end{lemma}
\begin{proof}
	Let $\I_\text{max}$ and $r_\text{max}$ be the set of ironing intervals and reserve price in quantile space of $\Rmax^\star$. Comparing $\Rmax^\star$ and $\Rmin^\star$ directly is difficult since the ironing intervals and reserve price may be quite different between the two. Instead, we will take the set $\I_\text{max}$ of ironing intervals $[a_i,b_i)$, and reserve quantile $r_\text{max}$ from $\Rmax^\star$ and use this to iron $\Rmin(q)$ at intervals $[a_i-2\epsilon - \frac1m,b_i- 2\epsilon - \frac1m)$, and set a reserve price of $r_\text{max}-2\epsilon-\frac1m$.
	
	If we compare $\Rmax^\star$ against this induced curve $\Rmin^{(\I_\text{max},r_\text{max})}$, we have an upper bound for the difference with respect to $\Rmin^\star$, since the latter is the optimal induced revenue curve and therefore pointwise higher. We can reason about ironing in quantile space without loss of generality since both $\Rmax$ and $\Rmin$ use the same function $\qest$ (alternatively think of it as ironing in value space based on $\Rmax$, the line through the origin and $\Rmax(q)$ intersects $\Rmin$ at $q-2\epsilon-\frac1m$).
	
	There are 3 cases to handle: 1) if $q$ falls in an ironing interval $[a_i,b_i)\in\I_\text{max}$, 2) if the quantile $q$ is higher than the reserve quantile $q\geq r_\text{max}$ of $\Rmax$, and 3) when neither of those cases apply. We start with the case for a $q$ where $\Rmax^\star(q) = \Rmax$, i.e. $q$ does not fall in an ironing interval and is smaller than the reserve quantile:
	
	\begin{align*}
		\Rmax^\star(q) &= q\cdot \qest\left(1-q+\epsilon + \frac1m\right) \\
		&\leq \left(q-2\epsilon-\frac1m\right)\cdot \qest\left(1-q+\epsilon + \frac1m\right) + \left(2\epsilon + \frac1m\right)\vmax\\
		&= \Rmin\left(q-2\epsilon-\frac1m\right) + \left(2\epsilon + \frac1m\right)\vmax\\
		&= \Rmin^\star\left(q-2\epsilon-\frac1m\right) + \left(2\epsilon + \frac1m\right)\vmax\\
		&\leq \Rmin^\star(q) + \left(2\epsilon + \frac1m\right)\vmax
	\end{align*}
	where the first inequality holds because $\qest(q)\leq \vmax$ for all $q$ and the last inequality follows because $\Rmin^\star$ is monotonically non-decreasing. Rearranging yields the claim for $q$ not in an ironing interval or reserved.
	
	The other cases follow similarly: consider the case where $q\geq r_\text{max}$.
	\begin{align*}
		\Rmax^\star(q) &= \Rmax\left(r_\text{max}\right)\\
		&= r_\text{max} \cdot \qest\left(1-r_\text{max} +\epsilon + \frac1m\right)\\
		&\leq \left(r_\text{max}-2\epsilon - \frac1m\right) \cdot \qest\left(1-r_\text{max} +\epsilon + \frac1m\right) + \left(2\epsilon + \frac1m\right)\vmax\\
		&= \Rmin\left(r_\text{max} - 2\epsilon - \frac1m\right) + \left(2\epsilon + \frac1m\right)\vmax\\
		&\leq \Rmin^\star\left(r_\text{max} - 2\epsilon - \frac1m\right) + \left(2\epsilon + \frac1m\right)\vmax\\
		&\leq \Rmin^\star\left(q\right) + \left(2\epsilon + \frac1m\right)\vmax
	\end{align*}
	where the last inequality holds because $r_\text{max} - 2\epsilon - \frac1m \leq r_\text{max} \leq q$ and $\Rmin^\star$ is non-decreasing.
	
	Finally the case for when $q$ falls in an ironing interval is analogous: $\Rmax^\star(q)$ is a convex combination of the end points of the ironing interval: $\Rmax^\star(q) = \frac{q-a_i}{b_i-a_i}\cdot \Rmax(a_i) + \left(1-\frac{q-a_i}{b_i-a_i}\right)\Rmax(b_i)$ and hence $$\Rmin^\star(q) \geq \frac{q-a_i}{b_i-a_i}\cdot \Rmin\left(a_i-2\epsilon - \frac1m\right) + \left(1-\frac{q-a_i}{b_i-a_i}\right)\Rmin\left(b_i-2\epsilon-\frac1m\right) + \left(2\epsilon + \frac1m\right)\vmax.$$
	And so $\Rmax^\star(q)-\Rmin^\star(q) \leq \left(2\epsilon+\frac1m\right)\vmax$ in all cases.	
\end{proof}

Theorem~\ref{thm:main-theorem} now follows by combining Lemmas~\ref{lem:tech-lemma} and~\ref{lem:loss-rev-curve}.
The additive loss the in expected revenue of Algorithm~\ref{alg:em} is at
most $3\epsilon\cdot  n\cdot \vmax$.  

\vspace{.5\baselineskip}

\begin{prevproof}{Theorem}{thm:main-theorem}
	By Lemma~\ref{lem:tech-lemma} we can express the total
        additive error of the expected revenue of an algorithm that
        yields  ironing intervals $\I_{alg}$ and reserve price $r_{alg}$ with respect to the optimal auction as:
	\begin{align*}
		\frac{
			Rev[F,\I_{opt},r_{opt}] - Rev[F,\I_{alg},r_{alg}]}{n} \leq \max_{q\in[0,1]}\left(\Ropt(q) - \Ralg(q) \right).
	\end{align*}	
	By Lemma~\ref{lem:loss-rev-curve}, Algorithm~\ref{alg:em} yields
	\begin{align*}
		\max_{q\in[0,1]}\left(\Ropt(q) - \Ralg(q) \right)\leq \left(2\epsilon + \frac1m\right)\vmax.
	\end{align*}
	The theorem follows.
\end{prevproof}

When valuations lie in $[1, H]$, so
that the optimal expected revenue is bounded below, 
these results easily imply a sample complexity upper bound
$O\left(\epsilon^{-2}\ln\epsilon^{-1}n^2\vmax^2\right)$
for learning (efficiently) a $(1-\eps)$-(multiplicative)
approximate auction.

In Section~\ref{sec:environments} we show how to extend these results from single-item auctions to matroid environments and position environments.

\section{Matroid and Position Environments}
\label{sec:environments}

In Lemma~\ref{lem:tech-lemma} we showed that for environments with optimal auctions in $\A_{\I,r}$, the additive loss of Algorithm~\ref{alg:em} can be bounded in terms of the error in estimating the ironed revenue curve.

In this section we show that Myerson's optimal auction can be
expressed in this way for more general single-parameter environments
with i.i.d.\ bidders as well: namely in position auctions and auctions
with matroid constraints. The theorem statement is as follows; the
proof follows from Fact~\ref{fact:position} and
Fact~\ref{fact:matroid} that show that the optimal auctions for those
environments are in $\A$ (i.e., have the form $A^{(\I,r)}$ for a
suitable choice of ironed intervals~$\I$ and reserve price~$r$).

\begin{theorem}
	For position and matroid auctions with $n$ i.i.d.\ bidders with values from unknown distribution $F$, $m$ i.i.d.\ samples from $F$, with probability $1-\delta$, the additive loss in expected revenue of Algorithm~\ref{alg:em} compared to the optimal expected revenue is at most $3\cdot n\cdot \vmax\cdot \sqrt{\frac{\ln 2\delta^{-1}}{2m}}$.
\end{theorem}

Our results stand in contrast to previous work on picking the optimal reserve price based on samples, for which a $\Omega(\frac{\log n}{\log \log n})$ lower bound is known in these settings \citep{Devanur14}. 

\subsection{Position Auctions}
A position auction \citep{V07} is one where the winners are given a position, and position $i$ comes with a certain quantity $x_i$ of the good. The canonical example is that of ad slot auctions for sponsored search, where the best slot has the highest click-through-rate, and subsequent slots have lower and lower click-through-rates. In an optimal auction, the bidder with the highest ironed virtual value gets the best slot, the second highest ironed virtual value the second slot, and so on.

\begin{fact}
	\label{fact:position}
	The optimal auction $A_{(pos)}^{opt}$ for position auctions can be expressed as an auction with ironing and reserve price in value space: $A_{(pos)}^{opt}\in \A$.
\end{fact}
\begin{proof}
	In the optimal auction, the bidder with the highest ironed virtual
	value is awarded the first position (with allocation $x_1$),
	the bidder with the second highest ironed virtual value the second
	position with $x_2$, and so on. Since the ironed virtual value
	are monotonically non-decreasing in the value of a bidder, and
	identical in ironing intervals, this can equivalently be
	described by an auction in $\A$. 
\end{proof}

\subsection{Matroid Environments}
In a matroid environment, the feasible allocations are given by matroid $\M=(E, I)$, where $E$ are the players and $I$ are independent sets. The auction can simultaneously serve only sets $S$ of players that form an independent set of the matroid $S\in I$. A special case of this is the rank $k$ uniform matroid, which accepts all subsets of size at most $k$, i.e. it is a $k$-unit auction environment.

In matroid environments, the ex-post allocation function $x_i(\bf b)$ and interim allocation function $y_i(q)$ are no longer the same for each player, e.g. imagine a player $i$ who is not part of any independent set, then $y(q)=0$ everywhere. So for matroid environments, the total additive loss is
\begin{align*}
	Rev[F,\I_{opt},r_{opt}] - Rev[F,\I_{alg},r_{alg}] &= \sum_{i=1}^n \left( -\int_0^1 \Ropt(q)y_i'(q) dq +\int_0^1 \Ralg(q)y_i'(q) dq\right)\\
&= ... \\
&= n\cdot \max_{q\in[0,1]}\left(\Ropt(q) - \Ralg(q) \right).
\end{align*}

We need to show that the optimal allocation can still be expressed in terms of an auction $A\in \A$.

\begin{fact}
\label{fact:matroid}
The optimal auction $A_{(mat)}^{opt}$ for matroid auctions can be
expressed as an auction with ironing and reserve price in value space:
$A_{(pos)}^{opt}\in \A$. 
\end{fact}
\begin{prevproof}{Fact}{fact:matroid}
	A property of matroids is that the following simple greedy algorithm yields the optimal solution:
	\begin{codebox}
		\Procname{$\proc{Greedy}(E,I)$}
		\li $S \leftarrow \emptyset$
		\li \While $\{i : i\not\in S \land S\cup \{i\}\in I \}\neq \emptyset$
		\li \Do $S \leftarrow \argmax \{ v_i : i\not\in S \land S\cup \{i\}\in I\}$
		\End
		\li \Return $S$
	\end{codebox}
	where $v_i$ is the (ironed virtual) value associated with bidder
	$i$. Since the order of largest values is the same for both
	virtual values and bids (up to ties), the allocation of the optimal auction
	is identical to the auction that irons on $\I_{opt}$ and has
	reserve price $r_{opt}$ (up to tie-braking); hence $A_{(mat)}^{opt}\in \A$. 
\end{prevproof}

\section{No-Regret Algorithm}
\label{sec:no-regret}

So far we assumed access to batch of samples before having to choose
an auction.
In this section we show that even if we start out without any samples,
but run the auction is a repeated setting, using past bids as samples
leads to a no-regret algorithm. 
The goal here is to achieve additive error
$o(T) \cdot O(\text{poly}(n,\vmax,\delta))$ --- the error can be
polynomial in all parameters except the time horizon $T$, for which it
should be strictly sublinear.
We show that for Algorithm~\ref{alg:no-regret} the total loss grows as $\widetilde O(\sqrt{T}\sqrt{n}\sqrt{\log(\delta^{-1})}\vmax)$ and hence results in a no-regret algorithm. 


\begin{algorithm}[t]
	\begin{codebox}
		\Procname{$\proc{No-Regret-Auction}(\delta, T)$}
		\li $\rhd$ Round 0:
		\li Collect a set of bids $\bf b$, run $\proc{VCG}(\bf b)$
		\li $X\leftarrow \bf b$
		\li \For round $t=1...T$ \Do
		\li Collect a set of bids $\bf b$
		\li $\proc{EmpiricalMyerson}(X, \delta/T, \bf b)$
		\li $X\leftarrow X \cup \bf b$
		\End
	\end{codebox}
	\caption{A no-regret algorithm for optimal auctions.}
	\label{alg:no-regret}
\end{algorithm}

Invoking Theorem~\ref{thm:main-theorem} with confidence
parameter~$\delta/T$ and taking a union bound, we have the following.

\begin{fact}
	\label{fact:confidence}
	With probability $1-\delta$, for all rounds simultaneously, each round $t\in [1,T]$ of Algorithm~\ref{alg:no-regret} has additive loss at most $3\sqrt{\frac{\ln(2T\delta^{-1})}{2nt}}\cdot n \cdot \vmax$.
\end{fact}

This leads to the following no-regret bound.
\begin{theorem}
	\label{thm:iid-no-regret}
	With probability $1-\delta$, the total additive loss of
        Algorithm~\ref{alg:no-regret} is $O(\sqrt{n}\sqrt{T\log T}
        \sqrt{\log\delta^{-1}}\cdot \vmax)$, which is
        $\widetilde{O}(T^{1/2})$ with respect to $T$.
\end{theorem}
\begin{proof}
	By Fact~\ref{fact:confidence} with probability $1-\delta$ for all rounds simultaneously, the additive loss for round $t$ is bounded by $3\sqrt{\frac{\ln(2T\delta^{-1})}{2nt}}\cdot n\cdot \vmax$. The loss of day 0 is at most $\vmax\cdot n$. The total loss can then be bounded by:
	\begin{align*}(n \cdot \vmax)\cdot \left(1+3 \sum_{t=1}^T \sqrt{\frac{\ln(2T\delta^{-1})}{2nt}}\right)
	\end{align*}
	We can rewrite the sum:
	\begin{align*}
		\sum_{t=1}^T \sqrt{\frac{\ln (2T/\delta)}{2nt}}	&= \sqrt{\frac{\ln (2T/\delta)}{2n}}\cdot \sum_{t=1}^T \sqrt{\frac{1}{t}}\\
		&\le \sqrt{\frac{\ln (2T/\delta)}{2n}}\cdot 2\sqrt{T}\\
		&= \sqrt{\frac{2T \ln (2T/\delta)}{n}}
	\end{align*}
	Hence the total loss is
	\begin{align*}
		(n \cdot \vmax)\cdot \left(1+3  \sqrt{\frac{2T \ln (4T/\delta)}{n}}\right) = O( \sqrt{n}\sqrt{T\log T} \sqrt{\log\delta^{-1}}\cdot \vmax);
	\end{align*}
	the dependence on $T$ is $O(\sqrt{T\log T})=\widetilde{O}(\sqrt{T})$.
\end{proof}

The bound of $O(\sqrt{T\log T})$ is almost tight, as there is a lower
bound of $\Omega(\sqrt{T})$ given by \cite{Cesa-Bianchi13}. Also note
that if we do not know $T$ a priori, we can use a standard doubling
argument to obtain the same asymptotic guarantee.

\bibliographystyle{plain}

\appendix
\section{Reduced Information Model}
\label{sec:information}


This appendix considers a reduced information model for i.i.d. bidders
where the auctioneer can only set a reserve price and then observes
the selling price; see also~\cite{Cesa-Bianchi13}.
The goal is to model an observer who can see the outcome of past
auctions, and perhaps submit a ``shill bid'' to set a reserve, but
cannot directly observe the bids.

We assume 
we observe $m$ samples from the second highest order statistic (so out
of $n$ i.i.d. bids, we see the second highest bid). 
%
We show that there are distributions such that, in order to get the
performance close to that of running Myerson in the reduced
information model, you first need to see at least an exponential
number of samples. This is far more than what would be needed with
regular valuation distributions, where only the monopoly price is
relevant.

\subsection{Lower Bound}

\begin{theorem}
	\label{thm:F2-lb}
	For any $\epsilon>0$, to obtain $\frac14-\epsilon$ additive loss compared to Myerson's auction, with constant probability you need $\Omega(2^n/n)$ samples from $F_{(2)}$.
\end{theorem}

Before proving the statement, let's think about what this means: it means that if all we observe are samples from $F_{(2)}$, for large $n$, we cannot hope to get a vanishing regret in a polynomial (in $n$) number of steps.
Moreover, on this distribution, Myerson obtains at most 2 expected revenue, so we lose $\frac18$ of the profit (see corollary after the proof).

\begin{figure}
	\centering
	\subfloat[Two distributions, $D_1$ and $D_2$ that are hard to distinguish using samples from $F_{(2)}$. The distributions agree on $q\in \lbrack0,1/2\rbrack$ and disagree elsewhere. ]{\includegraphics[width=0.45\textwidth]{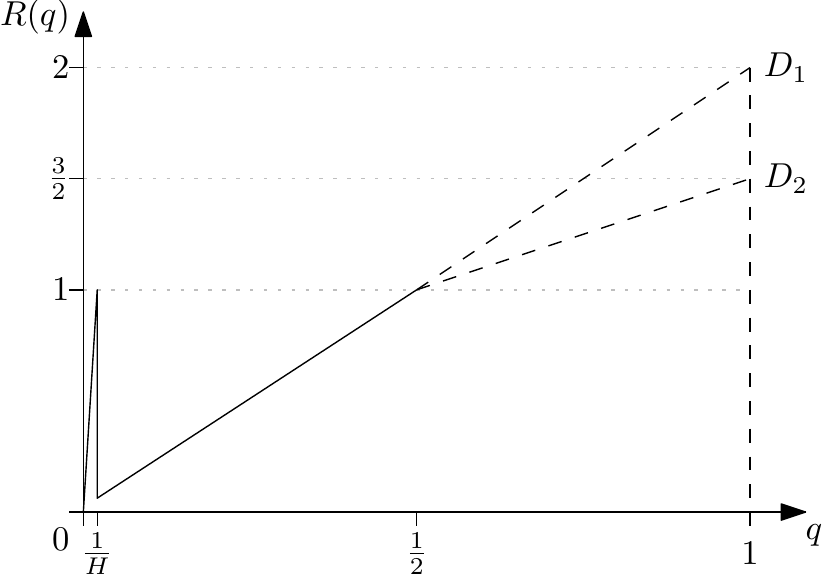}\label{fig:F2-lb}}
	\qquad
	\subfloat[Consider ironing on $[c,H)$. The top shaded area indicates the loss with respect to $D_1$ and the bottom shaded area represents the loss with respect to $D_2$.]{\includegraphics[width=.45\textwidth]{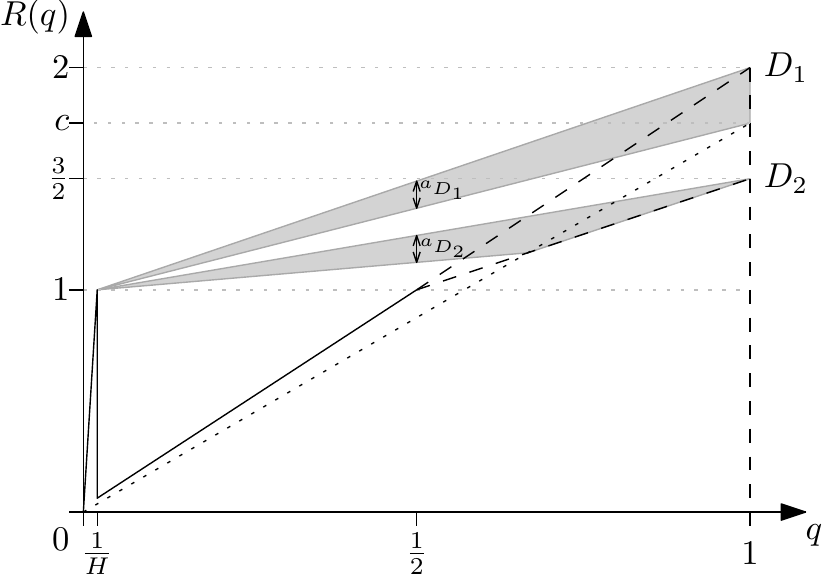}\label{fig:F2-lb-2}}
	\caption{Lower bound example for optimal auctions using samples from $F_{(2)}$.}
\end{figure}

\begin{prevproof}{Theorem}{thm:F2-lb}
	We'll give 2 distributions, $D_1$ and $D_2$, such that we need $\Omega(2^n)$ samples from $F_{(2)}$ to differentiate between the two. Moreover, we'll show that every auction that approximates the optimal auction for either $D_1$ or $D_2$, incurs an additive loss of at least $\frac12(1-\frac1{\sqrt2})$ on one of the two.
	
	We define $D_1$ and $D_2$ by their quantile function:
	
	\begin{align*}
	D^{-1}_1(q) &= \begin{cases}
	H & \text{for } q \leq \frac1H\\
	2 & \text{otherwise}
	\end{cases}\\
	D^{-1}_2(q) &= \begin{cases}
	H & \text{for } q \leq \frac1H\\
	2 & \text{for } \frac1H < q \leq \frac12\\
	1+\frac1{2q} & \text{otherwise}.
	\end{cases}
	\end{align*}
	
	Here $H\gg n$ is a sufficiently large constant. The two distributions agree on $q\in[0,\frac12]$. For $q\in (\frac12,1]$ they differ, and the effect on the revenue curve can be seen in Figure~\ref{fig:F2-lb}.
	
	To complete the proof we need to show 2 things: 1) that differentiating between $D_1$ and $D_2$ based on samples from $F_{(2)}$ requires $\Omega(2^n/2)$ samples, and 2) that no auction can simultaneously approximate the optimal solution for both.
	
	The first aspect of this is straightforward. Whenever we observe a sample from $F_{(2)}$ from the top half quantile: $q\in[0,\frac12]$, we get no information, since this is the same for both distributions. So a necessary condition for differentiating between $D_1$ and $D_2$ is if we observe a sample from $q\in(\frac12,1]$. For this to happen, we need that out of $n$ draws, $n-1$ draws are in the bottom quantile. This happens with probability $\frac{n}{2^{n-1}}$, so after $\frac{2^n}{n}$ samples, with probability approximately $1-\frac1e$ we haven't seen a sample that will differentiate the two distributions.
	
	Now we will show that until the moment when you can differentiate between $D_1$ and $D_2$, there is no way to run an auction that performs well for both. Note that the optimal auction for $D_1$ is to iron $[2,H)$ and the optimal auction for $D_2$ is to iron $[\frac32,H)$. Neither has a reserve price. The set of all auctions that could potentially work well for either, is the set of auctions that irons $[c,H)$ (see Figure~\ref{fig:F2-lb-2}).
	
	We'll only count the loss we incur in the range $[0,\frac12]$, take $H\rightarrow \infty$ and we see that if $a$ is the height of the shaded region (of $D_1$ resp. $D_2$) at $q=\frac12$, then its loss with respect to the optimal auction is:
	
	\begin{align*}
	n\cdot\int_0^{\frac12}2aq\cdot(n-1)(1-q)^{n-2} dq &=  \left[2a n q (1-q)^{n-1}\right]_0^{1/2} - 2an\int_0^{\frac12}(1-q)^{n-1} dq\\
	&=  2an\left(\frac{1}{2}\right)^{n-1} - 2an\left[\frac1n (1-q)^{n}\right]_0^{1/2}\\
	&=  2an\left(\frac{1}{2}\right)^{n-1} - 2an\frac1n \left(\frac12\right)^{n} + 2an\frac1n\\
	&=  2an\left(\frac{1}{2}\right)^{n-1}\left(1 - \frac1{2n}\right) + 2a\\
	&\geq 2a
	\end{align*}
	
	This means, that if we can show that for all choices of $c\in[3/2,1]$ at least one of $D_1, D_2$ has a large $a$, we are done. The closer $c$ is to $2$, the smaller $a_{D_2}$ is, the closer $c$ is to $3/2$, the smaller $a_{D_1}$, so we'll balance the two and show that neither is small enough.
	
	Finding $a_{D_1}$ in terms of $c$ is easy enough: $a_{D_1}=\frac{2-c}{2}$. $a_{D_2}$ is slightly harder: we first need to find the intersection of $q\cdot c$ and $R_{D_2}(q)$:
	
	\begin{align*}
	\frac12 + q &= qc\\
	(c-1)q &= \frac12\\
	q &= \frac1{2(c-1)}
	\end{align*}
	
	So the lower part of the shaded area of $D_2$ is a line that passes through the points $(0,1)$ and $(\frac{1}{2(c-1)},\frac{c}{2(c-1)})$. The line segment that connects the two is given by the equation $1+q(2-c)$ hence at $q=\frac12$ it is $2-c/2$, therefore $a_{D_2}=\frac54-2+\frac c2=\frac c2 - \frac34$.
	
	Setting $a_{D_1}=a_{D_2}$ yields $c=\frac74$ with $a_{D_1}=a_{D_2}=\frac18$. Therefore, for any auction that we decide to run will have additive loss of $2a=\frac14$ on one of the two distributions, and we need $\Omega(2^n/n)$ samples to decide which distribution we are dealing with.
\end{prevproof}

\begin{corollary}
	For any $\epsilon>0$, to obtain a $(\frac78+\epsilon)$-approximation (multiplicative) to Myerson's auction, with constant probability you need $\Omega(2^n/n)$ samples.
\end{corollary}
\begin{proof}
	We can upper bound the revenue curve by the constant $2$, to show that the optimal auction cannot have expected revenue more than $2$. Since we have additive loss of $1/4$, the multiplicative approximation ratio is at most $\frac78$.
\end{proof}
\end{document}